\documentclass[
final
]{dmtcs-episciences}

 \usepackage[square,numbers]{natbib}
\usepackage{amsthm}
\usepackage[utf8]{inputenc}
\usepackage{subfigure}
\newtheorem{theorem}{Theorem}[section]
\newtheorem{lemma}[theorem]{Lemma}
\def\ZZ{\mathbb{Z}}

\author{Yasir Ali\affiliationmark{1}
  \and Asma Javaid\affiliationmark{2}
}
\title[Stability with Strictly Increasing Valuations]{Pairwise Stability in Two Sided Market
 \\
with Strictly Increasing Valuation Functions}
\affiliation{
College of Electrical and Mechanical Engineering, National
University of Sciences and Technology, Rawalpindi 46070,
Pakistan\\
  School of Natural Sciences, National University
of Sciences and Technology, H-12, Islamabad, Pakistan}
\received{2016-3-17}%
\revised{2017-3-17}%
\accepted{2017-03-19}
\begin{document}
\publicationdetails{19}{2017}{1}{10}{1522}
\maketitle
\begin{abstract}
  This paper deals with two-sided matching market with two disjoint
sets, i.e. the set of buyers and the set of sellers. Each seller can
trade with at most with one buyer and vice versa. Money is
transferred from sellers to buyers for an indivisible goods that
buyers own. Valuation functions, for participants of both sides, are
represented by strictly increasing functions with money considered
as discrete variable. An algorithm is devised to prove the existence
of stability for this model.
\end{abstract}
\section{Introduction}
Over the last few decades, numerous scholars have carried out
research pertaining to two-sided matching problem. In a two sided
matching problem, the set of participants are divided into two
disjoint sets, say $U$ and $V$. Each participant ranks a participant
of other set in order of preferences. Main objective of two-sided
matching problem is formation of partnership between the
participants of $U$ and $V$. A matching $X$, is  one-to-one
correspondence between the participant of one set to the participant
of other set. Main requirement in a two-sided matching problems is
that of stability of matchings. A matching is stable if all
participants have acceptable partners and there does not exist a
pair that is not matched but prefers each other to their current
partners.

The concept of finding two-sided stable matching was first given by
Gale and Shapley \cite{GS} in their paper \textit{``College
Admissions and the Stability of Marriage''}. In the course of
presenting an algorithm for matching applicants to college places,
they introduced and solved the stable marriage problem. This problem
deals with two disjoint sets of participants $U$ and $V$. Each
participant of these sets submits a preference list ranking a subset
of other set of participants in order of preference. The aim is to
form a one-to-one matching $X$ of the participants such that no two
participants would prefer each other to their partner in $X$. The
authors used their solution to this problem as a basis for solving
the extended problem where one of the sets consists of college
applicants, and the other consists of colleges, each of which has a
quota of places to fill. An important feature of their model is that
no negotiations are allowed among the participant of both sets. This
shows that participants in their model are rigid. Many additional
variants of the stable marriage problem have been discussed in the
literature.  Gusfield and  Irving \cite{GuIr} published a book that
covers many variants of original stable marriage problem such as the
preferences of agents may include ties, incomplete preferences,
weighted edges as well as non-bipartite versions such as roommate
problem.

Shapley and Shubik \cite{SS} presented the one-to-one buyer seller
model known as \textit{``assignment game''}. In their model,
participants are flexible because monetary transfer is permitted
among participants of both sets. Each participant on one side can
supply exactly one unit of some indivisible good and exchange it for
money with a participant from the other side whose demand is also
one unit. Shapley and Shubik \cite{SS} showed that the core of the
game is a non-empty complete lattice, where the core is defined as
the set of un-denominated outcomes. The core in their model is a
solution set based upon a linear programming formulation of the
model \cite{SS}.

After this, two-sided matchings have been studied extensively.
Different approaches have been made by many researchers in which
they generalize the marriage model of  Gale and Shapley \cite{GS}
and assignment game of Shapley and Shubik \cite{SS}. Main aim of
these researchers was to find common result for both of \cite{GS}
and \cite{SS} models in a more general way. Eriksson and Karlander
\cite{EriKar} and Sotomayor \cite{Sot2000} presented the hybrid
models. These models are the generalization of the discrete marriage
model \cite{GS} and continuous assignment game \cite{SS}. Existence
of stable outcome and the core is discussed in \cite{EriKar,
Sot2000}. Farooq \cite{Far2008} presented a one-to-one matching
model in which he identified the preferences of participants by
strictly increasing linear functions. He proposed an algorithm to
show the existence of pairwise stable outcome in his model by taking
money as a continuous variable. His model includes the marriage
model of Gale and Shapely \cite{GS}, assignment game of Shapely and
Shubik \cite{SS} and Erikson and Karlander \cite{EriKar} hybrid
model as special cases. The motivation of our work from the stable
matching literature is the model of Ali and Farooq \cite{ya2010}.
Ali and Farooq \cite{ya2010} presented a one-to-one matching model
by taking money as a discrete variable in linear increasing
function. They designed an algorithm to show that pairwise stable
outcome always exists. The complexity of Ali and Farooq's
\cite{ya2010} algorithm depends on the size of those intervals where
prices fall. Our model is the generalized form of Ali and Farooq
\cite{ya2010} model. We consider the preferences of participants by
general increasing function and designed an algorithm to find a
pairwise stable outcome in our model.

This paper is organized as follows. Section 2 describes of our model
briefly. Section 3 gives the Sequential Mechanisms for buyer and
seller. Section 4 describes the supple and demand characterization
of stable matching. We devise an algorithm which finds a stable
outcome in our model in Section 5. In Section 6, we discuss the main
result of our model.
\section{The Model Description}
The matching market under consideration consists of two types of
participants one type of participants are sellers and second type of
participants are buyers. Here $U$ and $V$ denote the sets of sellers
and buyers, respectively. Throughout in this paper, we model
matching markets as trading platforms where buyers and sellers
interact. Moreover, each buyer as well as seller can trade with at
most one participant on the other side of the market at a particular
time.  The negotiation and side payments between participants of
both sides are allowed. Naturally, each participant wants to gain as
much profit as possible from his/her partner. Let $E=U \times V$
denotes the set of all possible pairs of seller-buyer. Also when
buyer and seller interact with each other in auction market they
have some upper and lower bounds of prices. We express these bounds
by vector $\underline{\pi}, \overline{\pi} \in \mathbb{Z}^{E}$ where
always $\underline{\pi}_{ij}\leq \overline{\pi}_{ij}$ for each
$(i,j)\in E$\footnote{The notation $\mathbb{Z}$ stand for set of
integers and notation $\mathbb{R}$ stand for set of real numbers.
The notation $\mathbb{Z}^E$ stands for integer lattice whose  points
are indexed by $E$.}. The price vector is denoted by $p$ and define
as $p=( {p}_{ij} \in \mathbb{Z} | (i,j) \in E)$. The price vector is
said to be feasible price vector if it satisfies
$\underline{\pi}\leq p\leq \overline{\pi}$ \footnote{For any two
vectors $x\in \mathbb{Z}^E$ and $y\in \mathbb{Z}^E$, we say that $x
\leq y$ if $x_{ij} \leq y_{ij}$ for all $(i,j)\in E$.}.

Since each participant has preferences over the participants of the
other set, so the preferences of sellers over buyers and buyers over
sellers is given by the valuation function $f_{ij}(x)$ and
$f_{ji}(-x)$ for each $(i,j)\in E$. Here $f_{ij}(x)$ denotes the
valuation of seller $i \in U$, when he or she trade with buyer $j
\in V$, and get an amount of money $x$ from buyer $j \in V$.
Similarly, $f_{ji}(-x)$ represents the valuation of buyer $j \in V$,
when he or she trade with seller $i \in U$, and pays an amount of
money $x$. These valuation functions are strictly increasing
functions of money \footnote{By strictly increasing function we mean
that for $x > y$ implies $f(x)> f(y)$.}.
\section{The Buyer Seller Sequential Mechanism}
\label{sec- The Buyer-Seller Sequential Mechanism}%
Since $f_{ij}(x)$ and $f_{ji}(-x)$ denote the preferences of
participants so if $f_{ij}(x)\geq 0$, then we say that seller $i$ is
ready to make a partnership with buyer $j$ if $j$ pays $i$ an amount
$x$ of money. Or we can say that buyer $j$ is acceptable to seller
$i$ an amount $x$ of money. Also if $f_{ji}(-x)\geq 0$, then we say
that buyer $j$ is ready to make a partnership with seller $i$ an
amount $x$ of money. If $f_{i_0j_0}(x_1)>f_{i_0j_1}(x_1)$, then we
can say that seller $i_0$ \textit{prefers} buyer $j_0$ to buyer
$j_1$ at money $x_1$ where $i_0 \in U $ and $j_0,j_1\in V$ and
$x_1\in \mathbb{Z}$. If $f_{j_0i_0}(-x_1)>f_{j_0i_1}(-x_1)$, then we
can say that $j_0$ \textit{prefers} $i_0$ to $i_1$ at money $x_1$
where $i_0, i_1 \in U $ and $j_0 \in V$ and $x_1\in \mathbb{Z}$. If
$f_{i_0j_0}(x_1)= f_{i_0j_1}(x_1)$, then seller $i_0$ is
\textit{indifferent} between $j_0$ and $j_1$ at money $x_1$. Also,
if $f_{j_0i_0}(-x_1) = f_{j_0i_1}(-x_1)$, then buyer $j_0$ is said
to be \textit{indifferent} between $i_0$ and $i_1$ at money $x_1$.
If $f_{ij}(x) = 0$, then seller $i$ is indifferent between the buyer
$j$ and himself at $x$. If an individual is not indifferent between
any two participants then the preferences of such individual are
called \textit{strict preferences}. In our model, preferences of the
participants are not strict because these are based on monetary
transfer and therefore, different functions may have same value for
two distinct values of money. If $f_{ji}(-x) = 0$ for some $x \in
\mathbb{Z}$, then buyer $j$ is indifferent between the seller $i$
and himself at $x$. Preferences of participants are not strict in
our model because the monetary transfer is allowed between
participants of both sets.
\section{The Supply and Demand Characterization of Stable Matchings}
\label{sec-The Supply and Demand Characterization of Stable
Matchings}%
This section describes the characteristic of an outcome for which it
would be stable. A subset $X$, of a set $E$, is called matching if
every agent appear at most once in $X$. A matching $X$ is said to be
pairwise stable if it is individually rational and is not blocked by
any buyer-seller pair. A 4-tuple $(X ; p, q, r)$ of a matching $X$
and a feasible price vector $p$ is said to be a
\textit{pairwise-stable outcome} if the following two conditions are
satisfied:
\begin{description}
\item[(p1)] $q \geq \bf{0}$ and
$r\geq \bf{0}$,
\item[(p2)] $f_{ij}(c) \leq q_i$ or
$f_{ji}(-c) \leq r_j $ for all $c \in [\underline{\pi}_{ij},
\overline{\pi}_{ij} ]_{\mathbb{Z}}$ and for all $(i,j)\in
E$\footnote{we define $[x,y]_{\mathbb{Z}} = \{a\in \mathbb{Z} \mid\
x\leq a \leq y\}$.},
\end {description}
where $(q, r)\in \mathbb{R}^{U}\times\mathbb{R}^{V}$ is defined by
\begin{eqnarray}
q_i & = & \left \{
\begin{array}{ll}f_{ij}(p_{ij}) & \mbox{if } (i,j)\in X \mbox{ for some } j\in V \\
0                & \mbox{otherwise}
\end{array}
\right.
\begin{array}{l}
\quad (i\in U),
\end{array}\label{payoff1}\\
r_j &=& \left \{
\begin{array}{ll}
f_{ji}(-p_{ij})  & \mbox{if } (i,j)\in X \mbox{ for some } i\in U\\
0                  & \mbox{otherwise}
\end{array}
\right.
\begin{array}{l}
\quad (j\in V).
\end{array}\label{payoff2}
\end{eqnarray}

Condition (p1) says that the matching $X$ is individually rational.
Condition (p2) means $(X; p, q, r)$ is not blocked by any
buyer-seller pair. A matching $X$ is said to be
\textit{pairwise-stable} if $(X; p, q, r)$ is pairwise-stable.

To show the existence of pairwise-stable outcome in the model
defined in Section 3, we first need to calculate price vector $p$
for each buyer-seller pairs. Since prices should be feasible and
$p_{ij} \in \mathbb{Z}$ for each  $(i,j)\in
E$\footnote{$\lfloor{x}\rfloor= \sup\{n\in \ZZ\mid x\geq n\}$.}, so
initially we define it by

\begin{equation}
\label{inip} p_{ij} = \left\{ \begin{array}{ll}
\overline {\pi}_{ij} & \mbox{if }  f_{ji}(-\overline {\pi}_{ij}) \geq 0\\
\max\left\{\underline {\pi}_{ij},
\left\lfloor{-f^{-1}_{ji}(0)}\right\rfloor \right. &
    \mbox{otherwise}.
 \end{array}
\right.
\end{equation}
Before describing the algorithm mathematically, we define few
subsets of set $E$ that help us to find a matching $X$ satisfying
condition (p1). Firstly, we define the subset $K_{0}$ and $T_{0}$ of
set $E$, that contain those buyer-seller pairs from the set $E$ that
are not mutually acceptable, as:
\begin{equation}\label{K0}
K_0 = \{ (i,j)\in E \mid f_{ji}(-p_{ij})< 0\},
\end{equation}
\begin{equation}\label{T0}
T_0 = \{ (i,j)\in E \mid f_{ij}(p_{ij})< 0\}.
\end{equation}
$K_0$ is the set of all those pairs where buyer is not ready to
trade with seller and $T_0$ is the set of all those pairs where
seller is not ready to trade with buyer. Now the set of mutually
acceptable buyer-seller pairs is defined as:
\begin{equation}\label{tile}
\widetilde E = E\setminus\{K_0 \cup T_0\}.
\end{equation}
Define $\tilde q_i$ for each $i\in U$, and $\widetilde E_P$ by
(\ref{til-q}) and (\ref{til-EP})
\begin{equation}\label{til-q}
\tilde q_i = \max \{f_{ij}(p_{ij})\mid (i,j)\in \widetilde E \}
\end{equation}
and
\begin{equation}\label{til-EP}
\widetilde E_P = \{(i,j)\in \widetilde E\mid f_{ij}(p_{ij})= \tilde
q_i\}.
\end{equation}
The maximum over an empty set is taken to be zero by definition.
Here the set $\widetilde E_P$ contains those buyer-seller pairs
which are mutually acceptable and the buyer is most preferred for
seller out of all acceptable buyers. We define a subset
$\widehat{E}_P$ of $\widetilde{E}_P$ by:
\begin{equation}\label{hatEP}
\widehat{E}_P = \{ (i,j) \in \widetilde{E}_P \mid
f_{ji}(-p_{ij})\geq r_j\}.
\end{equation}
Initially, since $r= \textbf{0}$, $\widehat{E}_P$ will coincide with
$\widetilde{E}_P$. However, in the further iterations of the
algorithm $\widehat{E}_P$ may be a proper subset of
$\widetilde{E}_P$.

Since we have no matching $X$ at the start of the algorithm, so
consider $\widetilde {V} = \emptyset$, where $\widetilde {V}$
denotes the set of matched buyers in $X$, that is,

\begin{equation}\label{til-V}
\widetilde V = \{j \in V\mid j \mbox{ is matched in } X \}.
\end{equation}
If $\widetilde {V} = \emptyset$, then there is no matched buyer in
matching $X$. At each step in the algorithm, the matching $X$ in the
bipartite graph $(U,V;\widehat{E}_P)$ must satisfies the following
conditions:
\begin{description}
\item[(s1)] $X$ matches all members of $\widetilde{V}$,
\item[(s2)]
$X$ maximizes $\sum\limits_{(i,j)\in X}f_{ji}(-p_{ij})$ among the
matchings that satisfy (s1).
\end{description}
Up to this point the outcome $(X; p, q, r)$ obviously satisfies the
condition (p1). To satisfy the condition (p2), we define the set $K$
of all those buyer-seller pairs that are mutually acceptable and the
buyer is most preferred to seller but the seller is unmatched in $X$
by
\begin{equation}\label{K}
K =\{(i,j)\in \widetilde E_P\mid i \mbox{ is  unmatched in } X \}.
\end{equation}
\begin{lemma}
If $K = \emptyset$, then matching $X$ is stable.
\end{lemma}
\begin{proof}
We know that a stable matching satisfy conditions (p1) and (p2). By
definition  $X\subseteq \widetilde E$ thus (p1) holds true. Suppose
that $K=\emptyset$ and on contrary suppose that (p2) dose not hold
true. This means that for some $(i,j)\in E$ there exists $c\in
[\underline\pi_{ij}, \overline \pi_{ij}]$ such that $f_{ji}(-c)>r_j$
and $f_{ij}(c)>q_i$. Initially, $r= \textbf{0}$ and
$f_{ji}(-p_{ij})\geq 0$ for $(i,j)\in E$, therefore, $p_{ij}\geq c$,
by (\ref{inip}). This means that $f_{ij}(p_{ij})\geq f_{ij}(c)>q_i$.
But $K=\emptyset$ implies that $f_{ij}(p_{ij})< \tilde q_i=q_i$,
which is a contradiction. This proves the assertion.
\end{proof}

If $K=\emptyset$ then there is no need to modify price vector $p$
and define further sets but if $K$ is not empty then we will modify
price vector, by preserving condition (p1). The new price vector
must also be feasible, that is, $\underline{\pi}_{ij}\leq \tilde
p_{ij} \leq \overline{\pi}_{ij}$ for each $(i,j)\in E$. Since we are
considering strictly increasing functions, therefore, we can find a
real number $m_{ij}^* \in \mathbb{R}^{{+}{+}}$ for each $(i, j)\in
K$, to modify price vector $p$, such that
\begin{equation}\label{m*}
f_{ji}(-(p_{ij}-m_{ij}^*)) = r_j.
\end{equation}
Since we are dealing with discrete prices so we will define an
integer $m_{ij}$ as follows:
\begin{equation}\label{mij}
m_{ij}= \max \left\{1, \lceil m_{ij}^{*}\rceil \right\}.
\end{equation}
Now, we have
\begin{equation}\label{mij int}
f_{ji}(-(p_{ij}-m_{ij})) \geq r_j,
\end{equation}
where $p_{ij}- m_{ij}$ is an integer and $m_{ij}$ is the minimum
positive integer that satisfies the above condition.\\
This means that
\[
 f_{ji}(-(p_{ij}-(m_{ij}-1))) \leq r_j.
 \]

Here the integer $m_{ij}$ for each $(i,j)\in K$ helps us in finding
the new price vector such that condition (p2) also satisfies. Now we
define
a subset $L$ of $K$ that contain those pairs from the set $K$ for
which modified price does not remain feasible.
\begin{equation}\label{tilL}
L=\{(i,j)\in K\mid p_{ij}-m_{ij}< \underline{\pi}_{ij}\}.
\end{equation}
The modified price vector $\tilde p$ must also be feasible and is
defined by:
\begin{equation}\label{til-p}
\tilde p_{ij}:=\left \{
\begin{array}{ll}
\max \{\underline {\pi}_{ij}, p_{ij}-m_{ij}\} & \mbox{if } (i,j)\in K\\
p_{ij} & \mbox{otherwise}
\end{array}\right.
\begin{array}{l}
\quad (i,j)\in E.
\end{array}
\end{equation}
We also define a subset $\widetilde T_0$ of $K$ by:
\begin{equation}\label{U2}
\widetilde T_0:=\{(i,j)\in K \mid  f_{ij}(\tilde p_{ij})< 0\}.
\end{equation}
\textbf{Remark:} Throughout in the algorithm, our modified price
vector will be decreasing and the size of matching $X$ will be
increasing. Also, the participants will change their preferences
according to new price vector.
\section{An Algorithm for Finding a Pairwise Stability}
In this section, we propose an algorithm for finding a pairwise
stable outcome for the model described in Section \ref{sec- The
Buyer-Seller Sequential Mechanism}.
\begin{description}
\item[Input:] Two disjoint and finite sets $U$ and $V$, the set of ordered pairs
$E= U\times V$, price vector $p \in \mathbb{Z}^E$, two vectors
$\underline{\pi} \in \mathbb{Z}^E$ and
$\overline{\pi}\in\mathbb{Z}^E$ where
$\underline{\pi}\leq\overline{\pi}$, general increasing functions .
\item[Output:] Vectors $(q,r)\in \mathbb{R}^{U}\times \mathbb{R}^{V}$, and  $p \in \mathbb{Z}^E $
 must satisfy $(p1)$ and $(p2)$.
\item[Step 0:]
Put $\widetilde V= \emptyset$ and $r = \bf{0}$. Initially define
$p$, $K_{0}$, $T_{0}$, $\widetilde E$, $\tilde q$, $\widetilde E_P$
and $\widehat{E}_P$ by (\ref{inip})$- $(\ref{hatEP}), respectively
and find a matching $X$ in the bipartite graph $(U,V;\widehat
{E}_P)$ satisfying (s1) and (s2). Define $r$, $\widetilde{V}$ and
$K$ by (\ref{payoff2}, (\ref{til-V}) and (\ref{K}), respectively.
\item[Step 1:]
If $K= \emptyset$ then define $q$ by (\ref{payoff1}) and stop.
Otherwise go to Step 2.
\item[Step 2:]
For each $(i,j)\in K$ calculate $m_{ij}$ by (\ref{mij}) and new
price vector $\tilde p$ by (\ref{til-p}). Define $L$ and $\widetilde
T_0$ by (\ref{tilL}) and (\ref{U2}), respectively and update $T_0$
by $T_0:=T_0 \cup \widetilde T_0$ and $K_0$ by $K_0:=K_0 \cup L$.
\item[Step 3:]
Replace price vector $p$ by $\tilde{p}$ and modify $\widetilde E$
by:
\begin{equation}\label{ft}
\widetilde E:=\widetilde E\setminus \{K_0 \cup T_0\}.
\end{equation}

Again define $\tilde q$ by (\ref{til-q}) and modify $\widetilde
E_P$, $\widehat E_P$ by (\ref{til-EP}) and (\ref{hatEP})
respectively, for the updated $p$ and $\widetilde E$. Find a
matching $X$ in the bipartite graph $(U, V; \widehat E_P)$ that
satisfies the conditions (s1) and (s2). Again define $r$,
$\widetilde{V}$ and $K$ by (\ref{payoff2}), (\ref{til-V}) and
(\ref{K}), respectively. Go to Step 1.
\end{description}
\section{Existence of Pairwise Stability} In this section, we
will show the existence of pairwise stability for this model. For
this purpose, we will show that the algorithm we have proposed
terminates and at termination it outputs a stable matching. We will
also give some other important results about the model and the
algorithm.

We will add prefixes $(old)*$ and $(new)*$ to sets/vectors/integers
before and after update, respectively, in any iteration of the
algorithm. The key result is Lemma \ref{max-intg} which will be
proved here using the assumption defined in equation (\ref{mij}).
%
\begin{lemma}\label{existence}
There exists a matching $X$ in the bipartite graph $(U, V; \widehat
E_P)$ that satisfy condition (s1) and (s2) in  each iteration of the
algorithm at Step 3.
\end{lemma}
\begin{proof}
The proof of the lemma is equivalent to show that $(old)X\subseteq
(new)\widehat E_P$ at Step 3, in each iteration. In each iteration
at Step 2 and at Step 3, we update vector $p$ and $\widetilde E$ by
$(\ref{til-p})$ and (\ref{ft}), respectively. As clear from
(\ref{til-p}) and (\ref{ft}), these modifications are done for
elements or/and subsets of $K$. As $K\cap(old)X=\emptyset$,
therefore, $(old)X\subseteq (new)\widehat E_P$.
\end{proof}

The following lemma represents the significance of $m_{ij}$ for each
$(i,j)\in K$ and explains that updated price is the maximum price at
which $(i,j)\in K$ can match.
\begin{lemma}\label{max-intg}
In each iteration of the algorithm at Step $3$, we have
$f_{ji}(-(p_{ij}-m_{ij}))\geq r_j$ for each $(i,j)\in K$.
Furthermore, if $f_{ji}(-(p_{ij}-m_{ij}))> r_j$ for some $(i,j)\in
K$ then $p_{ij}- m_{ij}$ is the maximum integer for which this
inequality holds.
\end{lemma}
\begin{proof}
Let $(i,j)\in K$ this means that $f_{ji}(-(old)p_{ij}) \leq r_{j}$.
At Step 2 we calculated an integer $m_{ij}$ by (\ref{mij}) for each
$(i,j) \in K$ with following property
\[
f_{ji}(-((old)p_{ij}-m_{ij})) \geq r_j.
\]
This proves the first part of the assertion.

Next, we prove the second part of the lemma that if
$f_{ji}(-((old)p_{ij}-m_{ij}))> r_j$ then $(old)p_{ij} - m_{ij}$ is
the maximum integer for which this holds. This can be proven by
showing that $m_{ij}$ is minimum positive integer for which
$f_{ji}(-((old)p_{ij}-m_{ij}))> r_j$ holds.

By (\ref{mij}), we have $m_{ij} \geq 1$. First we consider the case
when $m^{*}_{ij} \leq 1$, that is, $m_{ij} = 1$ by (\ref{mij}). 
For this case the result holds trivially as $m_{ij} = 1$ is minimum
positive integer. Now, consider when $m_{ij}^{*}
> 1$. We  $m_{ij}^{*}$ is a real number for which we have
\[
r_{j} = f_{ji}(-(old)p_{ij} + m_{ij}^{*})
\]
As we are dealing with strictly increasing function, therefore, for
any real number $\delta > 0$, we have
\begin{equation}\label{li}
f_{ji}(-((old)p_{ij} - (m_{ij}^{*} + \delta)))
> r_j
> f_{ji}(-((old)p_{ij} - (m_{ij}^{*} - \delta))).
\end{equation}
Since
\begin{equation}\label{ht}
m_{ij} = \lceil m_{ij}^{*}\rceil \geq m_{ij}^{*}.
\end{equation}
By (\ref{li}), $m^*_{ij}$ is minimum positive real number for which
$r_{j} < f_{ji}(-(old)p_{ij} + m_{ij}^{*})$ and by (\ref{ht}),
$m_{ij}$ is minimum positive integer for which
$f_{ji}(-(p_{ij}-m_{ij}))> r_j$. Thus $(old)p_{ij} - m_{ij}$ is the
maximum integer for which $f_{ji}(-(p_{ij}-m_{ij}))> r_j$ holds.
\end{proof}

For $(i,j)\in K$, we update price vector by (\ref{til-p}). There is
a possibility that $(new)p_{ij}$ does not remain feasible, that is,
$p_{ij}<\underline\pi_{ij}$. To maintain the feasibility in such
cases we have the following result.
\begin{lemma}\label{re-la02}
For each $(i,j)\in L$ we have $(new) p_{ij} =\underline \pi_{ij}$
and $f_{ji}(-(new) p_{ij}) \leq (old)r_j$ , where $L$ is defined at
Step 2.
\end{lemma}
The proof of the Lemma~\ref{re-la02} follows by using (\ref{tilL})
and (\ref{til-p}).
%

The following lemma describes the important features of our
algorithm. The results of these lemma will be used to show that the
algorithm will terminate after finite number of iterations.
\begin{lemma}\label{la2}
In each iteration of the algorithm, following hold:
\begin{description}
\item[(i)]
If $L\neq\emptyset$ or $\widetilde T_0\neq\emptyset$ at Step 2 then
$\widetilde E$ reduces at Step 3. Otherwise $\widetilde E$ will
remain the same.
\item[(ii)]
The vector $p$ decreases or remains same. In particular, if
$K\setminus\{L\cup \widetilde T_0\}\not=\emptyset$ at Step 2 then
$p_{ij}$ decreases at Step 3 for all $(i,j)\in K\setminus \{L\cup
\widetilde T_0\}$.
\item[(iii)]
The vector $r$ increases or remains  same.
\end{description}
\end{lemma}
\begin{proof}
\begin{description}
\item[(i)]At Step 0, $\widetilde E$ is given by (\ref{tile}) and it
is updated by (\ref{ft}), at Step 3. At Step 2 we updated
$K_0=K_0\cup L$ and $T_0= T_0\cup \widetilde T_0$. According to
(\ref{ft}), $\widetilde E$ will reduce if $L\neq\emptyset$ or
$\widetilde T_0\neq\emptyset$ at Step 2. If both $L$ and $\widetilde
T_0$ are empty the $\widetilde E$ will remain unchanged by
(\ref{ft}).
\item[(ii)] Initially, $p$ is set by (\ref{inip}) and in each
iteration it is updated by (\ref{til-p}). It is easy to see that for
$(i,j)\in K$, $\tilde p_{ij}\leq p_{ij}$. Here the equality may hold
for $(i,j)\in L\cup \widetilde T_0$.
\item[(iii)] At the start of the algorithm we set $r=\bf 0$. We
modified $r$ by (\ref{payoff2}) afterwards. In each iteration,
matching $X$ satisfies condition $(s1)$ this means that
$(old)\widetilde V\subseteq(new)\widetilde V$. Also,
$(new)p\leq(old)p$ by part (ii) of Lemma \ref{la2}. Thus
$(new)r_j=f_{ji}(-(new)p_{ij})\geq (old)r_j$, for $j\in
(old)\widetilde V$, as matching $X$ also satisfies $(s2)$. Moreover,
$(old)r_j=(new)r_j=0$ for each $j\in V\setminus \widetilde V$.
Therefore, vector $r$ either remains the same or increases.
\end{description}
\end{proof}

To show that our algorithm produces a stable matching is not
possible without proving that our algorithm will terminate after
some iterations.
\begin{theorem}\label{termination}
The algorithm terminates after finite number of iterations.
\end{theorem}
\begin{proof}
Termination of the algorithm depends upon set of mutually acceptable
pairs and price vector $p$. By the Lemma \ref{la2} part (i),
$\widetilde E$ reduces when either $L\neq\emptyset$ or
$T_0\neq\emptyset$ or remains the same. This case is possible at
most $|E|$ times.

If $L=T_0=\emptyset$ then, by part (ii) of Lemma \ref{la2}, $p_{ij}$
decreases for each $(i,j)\in K$. Otherwise, $p$ remains unchanged.
As we know that $p$ is bounded and discrete, therefore, it can be
decreased a finite number of time. This proves that in either case
our algorithm terminates after a finite number of iterations.
\end{proof}

This is the most important result which establishes the existence of
pairwise stability for our model.
\begin{theorem}\label{st}
The outcome $(X; p, q, r)$ must satisfies the condition (p1) and
(p2) if algorithm terminates.
\end{theorem}
\begin{proof}
We know that $X\subseteq \widetilde E$. Initially $\widetilde E$ is
defined by (\ref{tile}) and afterwards it is updated by (\ref{ft})
at Step 3 in each iteration. Thus $f_{ij}(p_{ij})$ and
$f_{ji}(-p_{ij})$ are non-negative for all $(i,j)\in \widetilde E$.
Therefore, $f_{ij}(p_{ij})\geq 0$ and $f_{ji}(-p_{ij})\geq 0$ for
all $(i,j)\in X$. This shows that the $X$ satisfies (p1) at
termination.

On contrary to (p2), assume that there exist $\alpha \in
[\underline\pi_{ij},\overline\pi_{ij}]$ and $(i,j)\in E$ such that
\[
f_{ij}(\alpha)>q_i \qquad \mbox{ and } \qquad f_{ji}(-\alpha)>r_j.
\]
If we take $p_{ij}<\alpha$ it yields
$f_{ji}(-p_{ij})>f_{ji}(-\alpha)>r_j$. But according to Lemma
\ref{max-intg}, $p_{ij}$ is the maximum integer for which this
inequality holds. Thus $p_{ij}<\alpha$ is not possible. Now consider
that $p_{ij}\geq \alpha$, which implies that
\begin{equation}\label{contra-eq}
f_{ij}(p_{ij})\geq f_{ij}(\alpha)> q_i.
\end{equation}
However, at termination we have $K=\emptyset$ means that $(i,j)\not
\in K$ and since $(i,j)$ are not matched, therefore,
$f_{ij}(p_{ij})<\tilde q_i=q_i$. A contradiction to
(\ref{contra-eq}). Thus (p2) holds when the algorithm terminates.
\end{proof}
\section{Conclusion}\label{sec-concl}
This paper presents a matching model where money is given in
integers. The preferences of participants are represented by general
increasing utility functions. Ali and Farooq \cite{ya2010} is a
special case of our model. We have given a constructive proof for
the existence of a pairwise stable outcome in our model. As a future
work it is important to consider problems concerning the structures
of pairwise stable outcomes in our model. It is well-known that
stable matchings forms a lattice. A similar approach can be found in
article \cite{Far08A} by Farooq \emph{et al.} It would be worthwhile
to prove the existence of stable outcome for many-to-many model with
such valuation functions by using the same mathematical apparatus.
Further, the complexity of our algorithm may depend on the length of
$[\underline \pi, \overline \pi]$. An interesting problem may be to
devise an algorithm with polynomial complexity in the number of
participants.

\end{document}